\documentclass[useAMS,usenatbib,referee]{biom}
\def\bSig\mathbf{\Sigma}

\usepackage{hyperref}

\usepackage{booktabs} 
\usepackage{xr}

\usepackage{amsmath}
\usepackage{amssymb}

\usepackage{caption}

\usepackage{booktabs}

\usepackage{graphicx}

\usepackage{centernot}

\usepackage{thmtools, thm-restate}

\DeclareMathOperator{\E}{\textnormal{\mbox{E}}}

\makeatletter
\newcommand*{\indep}{%
  \mathbin{%
    \mathpalette{\@indep}{}%
  }%
}
\newcommand*{\nindep}{%
  \mathbin{%                   % The final symbol is a binary math operator
    \mathpalette{\@indep}{\not}% \mathpalette helps for the adaptation
  }%
}
\newcommand*{\@indep}[2]{%
  \sbox0{$#1\perp\m@th$}%        box 0 contains \perp symbol
  \sbox2{$#1=$}%                 box 2 for the height of =
  \sbox4{$#1\vcenter{}$}%        box 4 for the height of the math axis
  \rlap{\copy0}%                 first \perp
  \dimen@=\dimexpr\ht2-\ht4-.2pt\relax
  \kern\dimen@
  {#2}%
  \kern\dimen@
  \copy0 %                       second \perp
} 
\makeatother

\makeatletter
\newcommand*{\addFileDependency}[1]{% argument=file name and extension
  \typeout{(#1)}
  \@addtofilelist{#1}
  \IfFileExists{#1}{}{\typeout{No file #1.}}
}
\makeatother

\newcommand*{\myexternaldocument}[1]{%
    \externaldocument{#1}%
    \addFileDependency{#1.tex}%
    \addFileDependency{#1.aux}%
}

\myexternaldocument{supp}

\title[Causally interpretable meta-analysis]{Efficient and robust methods for causally interpretable meta-analysis: transporting inferences from multiple randomized trials to a target population}

\author{Issa J. Dahabreh$^{1-3,*}$\email{idahabreh@hsph.harvard.edu}, 
Sarah E. Robertson$^{1,2}$,
Lucia C. Petito$^{4}$,
Miguel A. Hern\'an$^{1,2,5}$,  \\
\textbf{and Jon A. Steingrimsson$^{6}$} \\
$^{1}$CAUSALab, Harvard T.H. Chan School of Public Health, Boston, MA \\
$^{2}$Department of Epidemiology, Harvard T.H. Chan School of Public Health, Boston, MA \\
$^{3}$Department of Biostatistics, Harvard T.H. Chan School of Public Health, Boston, MA  \\
$^{4}$Department of Preventative Medicine, Feinberg School of Medicine, \\ Northwestern University, Chicago, IL \\
$^{5}$Harvard-MIT Division of Health Sciences and Technology, Boston, MA \\
$^{6}$Department of Biostatistics, School of Public Health, Brown University, Providence, RI}

\begin{document}

%  This will produce the submission and review information that appears
%  right after the reference section.  Of course, it will be unknown when
%  you submit your paper, so you can either leave this out or put in 
%  sample dates (these will have no effect on the fate of your paper in the
%  review process!)

\date{{\it Received January} 2021. {\it Revised February} 2022.  {\it
Accepted XXX} XXX.}

%  These options will count the number of pages and provide volume
%  and date information in the upper left hand corner of the top of the 
%  first page as in published papers.  The \pagerange command will only
%  work if you place the command \label{firstpage} near the beginning
%  of the document and \label{lastpage} at the end of the document, as we
%  have done in this template.

%  Again, putting a volume number and date is for your own amusement and
%  has no bearing on what actually happens to your paper!  

\pagerange{\pageref{firstpage}--\pageref{lastpage}} 
\volume{XXX}
\pubyear{2022}
\artmonth{XXX}

%  The \doi command is where the DOI for your paper would be placed should it
%  be published.  Again, if you make one up and stick it here, it means 
%  nothing!

\doi{10.1111/j.1541-0420.2005.00454.x}

%  This label and the label ``lastpage'' are used by the \pagerange
%  command above to give the page range for the article.  You may have 
%  to process the document twice to get this to match up with what you 
%  expect.  When using the referee option, this will not count the pages
%  with tables and figures.  

\label{firstpage}

%  put the summary for your paper here

\begin{abstract}
We present methods for causally interpretable meta-analyses that combine information from multiple randomized trials to estimate potential (counterfactual) outcome means and average treatment effects in a target population. We consider identifiability conditions, derive implications of the conditions for the law of the observed data, and obtain identification results for transporting causal inferences from a collection of independent randomized trials to a new target population in which experimental data may not be available. We propose an estimator for the potential (counterfactual) outcome mean in the target population under each treatment studied in the trials. The estimator uses covariate, treatment, and outcome data from the collection of trials, but only covariate data from the target population sample. We show that it is doubly robust, in the sense that it is consistent and asymptotically normal when at least one of the models it relies on is correctly specified. We study the finite sample properties of the estimator in simulation studies and demonstrate its implementation using data from a multi-center randomized trial.
\end{abstract}

%  Please place your key words in alphabetical order, separated
%  by semicolons, with the first letter of the first word capitalized,
%  and a period at the end of the list.
%


%  As usual, the \maketitle command creates the title and author/affiliations
%  display 

\maketitle

%  If you are using the referee option, a new page, numbered page 1, will
%  start after the summary and keywords.  The page numbers thus count the
%  number of pages of your manuscript in the preferred submission style.
%  Remember, ``Normally, regular papers exceeding 25 pages and Reader Reaction 
%  papers exceeding 12 pages in (the preferred style) will be returned to 
%  the authors without review. The page limit includes acknowledgements, 
%  references, and appendices, but not tables and figures. The page count does 
%  not include the title page and abstract. A maximum of six (6) tables or 
%  figures combined is often required.''

%  You may now place the substance of your manuscript here.  Please use
%  the \section, \subsection, etc commands as described in the user guide.
%  Please use \label and \ref commands to cross-reference sections, equations,
%  tables, figures, etc.
%
%  Please DO NOT attempt to reformat the style of equation numbering!
%  For that matter, please do not attempt to redefine anything!

\section{Introduction}
\label{s:intro}

When examining a body of evidence that consists of multiple trials, decision makers are typically interested in learning about the effects of interventions in some well-defined \emph{target population}. In other words, they are interested in \emph{synthesizing the evidence} across trials and \emph{transporting causal inferences} from the collection of trials to a target population in which further experimentation may not be possible. Typically, each trial samples participants from a different underlying population, by recruiting participants from centers with different referral patterns or in different geographic locations. The goal of evidence synthesis in this context is to use the information from these diverse trials to draw causal inferences about the target population, accounting for any differences between the target population and the populations underlying the trials.

``Meta-analysis,'' an umbrella term for statistical methods for synthesizing evidence across multiple trials \citep{cooper2009handbook}, traditionally focuses on modeling the distribution of treatment effects (effect sizes) across studies or on obtaining unbiased and minimum variance summaries of data from multiple trials \citep{higgins2009re, rice2018re}. Standard meta-analysis methods produce estimates that do not have a clear causal interpretation outside of the sample of participants enrolled in the trials because the estimates do not pertain to any well-defined target population \citep{dahabreh2020toward}. Recent work on ``generalizability'' and ``transportability'' has considered methods for extending causal inferences to a target population, when data are available from a single randomized trial \citep{westreich2017, rudolph2017, dahabreh2020transportingStatMed}. When multiple trials are available, methods have been proposed for assessing case-mix heterogeneity in individual patient data meta-analyses without specifying or using data from a target population of substantive interest  \citep{vo2019novel, vo2021assessing}. 

Here, we propose methods for \emph{causally interpretable meta-analysis} that can be used to extend inferences about potential (counterfactual) outcome means and average treatment effects from a collection of randomized trials to a target population in which experimental data may not be available \citep{dahabreh2020toward}. We consider identifiability conditions, derive implications of the conditions for the law of the observed data, and obtain identification results. We then propose novel estimators for potential outcome means and average treatment effects in the target population. The estimators use data on baseline covariates, treatments, and outcomes from the collection of trials, but only data on baseline covariates from the sample of the target population. We show that the estimators are doubly robust, in the sense that they remain consistent and asymptotically normal provided at least one of the two working models on which they rely is correctly specified. Last, we study the finite sample properties of the estimators in simulation studies and demonstrate the implementation of the methods using data from a multi-center trial of treatments for hepatitis C infection.

%%%%%%%%%%%%%%%%%%%%%%%%%%%%%%%%%%%%%%%%%%%%%%%%%%%%%%%%%%%%%%%%%%%%%%%%%%%%%%
\section{Data, sampling scheme, and causal estimands}\label{sec:data_causal_quant}
%%%%%%%%%%%%%%%%%%%%%%%%%%%%%%%%%%%%%%%%%%%%%%%%%%%%%%%%%%%%%%%%%%%%%%%%%%%%%%

Suppose we have data from a collection of randomized trials $\mathcal S$, indexed by $s = 1, \ldots, m$. For each trial participant we have information on the trial $S$ in which they participated, treatment assignment $A$, baseline covariates $X$, and outcomes $Y$. We assume that the same finite set of treatments, $\mathcal A$, has been compared in each trial of the collection (extensions to cases where only a subset of treatments are evaluated in each trial are straightforward but require more cumbersome notation). From each trial $s\in \mathcal S$, the data are independent and identically distributed random tuples $(X_i, S_i = s, A_i, Y_i)$, $i = 1, \ldots, n_s$, where $n_s$ is the total number of randomized individuals in trial $s$. 

We also obtain a simple random sample from the target population (of individuals not participating in the trials), in which treatment assignment may not be under the control of investigators (e.g., individuals may self-select into treatment). We use the convention that $S = 0$ for individuals from the target population. The data from the sample of the target population consists of independent and identically distributed tuples $(X_i, S_i = 0,A_i,Y_i)$, $i = 1, \ldots, n_0$, where $n_0$ is the total number of individuals sampled from the target population. The total number of observations from the trials and the sample of the target population is $n = \sum_{s=0}^{m} n_s$. We define a new random variable $R$, such that $R=1$, if $S \in \mathcal S$; and $R=0$, if $S = 0$.

Informally, we assume that the observations from the trials in the collection $\mathcal S$ and from the target population $S=0$ are Bernoulli-type (independent random) samples \citep{breslow2007weighted, saegusa2013weighted} from (near-infinite) underlying populations \citep{robins1988confidence}. Specifically, we assume that the sample from the target population is representative of a well-specified population of substantive interest. In contrast, we do not require that the trial samples are obtained through a formal sampling process; rather, that the investigators are willing to model the data from each trial \emph{as if} sampled from a hypothetical underlying population. This modeling choice is consistent with the fairly standard super-population approach to the analysis of individual randomized trials \citep{robins1988confidence}. Though alternative frameworks (e.g., randomization-based inference) can be appealing in some cases, we find the super-population approach attractive when the goal is to extend inferences from a collection of trials  to a new target population \citep{dahabreh2019commentaryonweiss}. Investigators conducting meta-analyses usually do not have control over the selection of participants of each trial or the relative sample size of different trials. Thus, we view the observations available for analysis as obtained by sampling the corresponding underlying populations with sampling fractions that are not under the control of the investigators and are unknown to them. In Section \ref{sec:sampling_model} we further formalize this sampling scheme; for now it suffices to say that all expectations and probabilities below are defined under the sampling scheme. Throughout, we use $f(\cdot)$ to generically denote densities. In Section \ref{sec:sampling_model} we discuss the implications of the sampling scheme for identifiability of these densities. 

To define causal estimands, let $Y^a$ denote the potential (counterfactual) outcome under intervention to set treatment to $a \in \mathcal A$ \citep{rubin1974, robins2000d}. We are interested in the potential outcome mean in the target population, $ \E [Y^a | R = 0]$ for each $a \in \mathcal A $, as well as the average causal effect, $\E [Y^a - Y^{a^\prime} | R = 0]$ for each pair of treatments $a \in \mathcal A$ and $a^\prime \in \mathcal A$. The treatments used in the target population need not be the same as the treatments used in the trials (e.g., some treatments may not be available outside experimental settings). Furthermore, as we show below, data on treatments and outcomes from the target population are not necessary for identification and estimation of the potential outcome means and average treatment effects of interest.

%%%%%%%%%%%%%%%%%%%%%%%%%%%%%%%%%%%%%%%%%%%%%%%%%%%%%%%%%%%%%%%%%%%%%%%%%%%%%%
\section{Identification}\label{sec:identification}
%%%%%%%%%%%%%%%%%%%%%%%%%%%%%%%%%%%%%%%%%%%%%%%%%%%%%%%%%%%%%%%%%%%%%%%%%%%%%%

\subsection{Identifiability conditions}

The following are sufficient conditions for identifying the potential outcome mean in the target population $\E[Y^a | R = 0]$.

\vspace{0.1in}
\noindent
\emph{A1. Consistency of potential outcomes:} if $A_i = a,$ then $Y^a_i = Y_i$, for every individual $i$ and every treatment $a \in \mathcal A$. 
\vspace{0.1in}

Implicit in condition \emph{A1} are assumptions that (i) there is no direct effect of participation in any trial ($R=1$) or participation in some specific trial ($S=s$) on the outcome \citep{dahabreh2019identification, dahabreh2019commentaryonweiss}; (ii) there are no multiple versions of treatment (or that treatment variation is irrelevant with respect to the outcome \citep{vanderWeele2009}); and (iii) there is no treatment-outcome interference \citep{rubin1986, rubin2010reflections}. These assumptions may be most plausible in the case of large pragmatic trials where fidelity to the assigned treatment can be high across settings. Note that consistency of potential outcomes is distinct from the notion of consistency used in so-called network meta-analyses; consistency of potential outcomes refers to the relationship between observed (factual) and potential (counterfactual) outcomes under each treatment, not the relationship between contrasts (``direct'' and ``indirect'', in meta-analytic parlance) of different treatments.

\vspace{0.1in}
\noindent
\emph{A2. Exchangeability over treatment $A$ in each trial:} for each trial $s \in \mathcal S$ and each $a \in \mathcal A$, $Y^a \indep A | (X, S = s)$. 
\vspace{0.1in}

Condition \emph{A2} is typically plausible because of randomization (marginal or conditional on $X$) in each of the trials . The condition is not equivalent to the condition $Y^a \indep A | (X, R = 1)$ because the trials may have different randomization ratios and trial participation may have direct effects on the outcome or share unmeasured common causes with the outcome (the last two possibilities, however, are precluded by conditions \emph{A1} and \emph{A4}, respectively). It is also worth noting that condition \emph{A2} would also hold if $\mathcal S$ was a collection of observational studies in which the covariates $X$ were sufficient to adjust for baseline confounding. Thus, our results can also apply to pooled analyses of observational studies, provided that background knowledge suggests ``adjustment'' for $X$ is sufficient to control confounding.

\vspace{0.1in}
\noindent
\emph{A3. Positivity of treatment assignment in each trial:} for each treatment $a \in \mathcal A$ and for each trial $s \in \mathcal S$, if $f(x, S = s) \neq 0$, then $\Pr[A = a | X = x, S = s]  > 0$. 
\vspace{0.1in}

Condition \emph{A3} is also plausible in marginally and conditionally randomized trials.

\vspace{0.1in}
\noindent
\emph{A4. Exchangeability over $S$:} for each $a \in \mathcal A$, $Y^a \indep S | X$. 
\vspace{0.1in}

In applied work, condition \emph{A4} will typically be a critical assumption connecting the randomized trials in the collection $\mathcal S$, between them and with the target population ($S = 0$). On its own, condition \emph{A4} is not testable; we will show, however, that it has testable implications when combined with conditions \emph{A1} and \emph{A2}). In practice, condition A4 will need to be examined in light of substantive knowledge and examining the impact of violations of the condition may require undertaking sensitivity analyses \citep{robins2000c}.  

\vspace{0.1in}
\noindent
\emph{A5. Positivity of trial  participation:} for each $s \in \mathcal S$, if $f(x, R = 0) \neq 0$, then $\Pr[S = s |X = x] > 0$.
\vspace{0.1in}

Informally, condition \emph{A5} states that covariate patterns in the target population can also occur in each trial. This condition only involves the observable data and thus is in principle testable. Formal examination of the condition, however, is challenging when $X$ is high dimensional \citep{petersen2012diagnosing}.

\subsection{Implications of the identifiability conditions}

We now explore some implications of the identifiability conditions. By Lemma 4.2 of \citet{dawid1979conditional}, condition \emph{A4} implies exchangeability over $S$ among the collection of trials, that is, 
\begin{equation}\label{eq:independencies1}
Y^a \indep S | X \implies  Y^a \indep S | (X, R = 1) , \mbox{ for every } a \in \mathcal A, 
\end{equation}
and also implies exchangeability of participants in the collection of trials and the target population,
\begin{equation}\label{eq:independencies2}
Y^a \indep S | X \implies Y^a \indep I(S = 0) | X \Longleftrightarrow Y^a \indep R | X.
\end{equation}
The first of these results will be useful to derive restrictions on the law of the observed data imposed by the identifiability conditions; the second result will be useful in obtaining identification results for causal quantities in the target population, using information from the collection of trials. 

By Lemma 4.3 of \citet{dawid1979conditional}, the result in \eqref{eq:independencies1} and condition \emph{A2} imply conditional exchangeability of potential outcomes over $(S, A)$ in the collection of trials, that is,
\begin{equation}\label{eq:independencies3}
    \begin{Bmatrix}
Y^a \indep A | (X, R=1, S) \\ 
Y^a \indep S | (X, R = 1) \\ 
\end{Bmatrix} \implies Y^a \indep (S, A) | (X, R = 1).
\end{equation}
Noting that $Y^a \indep (S, A) | (X, R = 1)$ implies $Y^a \indep S  | (X, R = 1, A = a)$, and using condition \emph{A1}, we obtain
\begin{equation}\label{eq:independencies4}
    Y \indep S | (X, R = 1, A = a).
\end{equation}
Thus, among individuals participating in any trial ($R=1$), the observed outcome $Y$ is independent of the trial random variable $S$, within treatment groups ($A = a$) and conditional on covariates $X$. Because this condition does not involve potential outcomes, it is testable using the observed data (e.g., using methods for comparing conditional densities). Furthermore, because $\{R = 0\} \Longleftrightarrow \{R=0,S=0\}$, we also obtain that $Y \indep S | (X, R, A = a)$. The independence condition in \eqref{eq:independencies4} implies that for every $a \in \mathcal A$ and every $x$ such that $f(x, S = 0) \neq 0$,
\begin{equation}\label{eq:observed_data_implications2}
  \E[Y | X = x, S = 1, A = a] = \ldots = \E[Y | X = x, S = m, A = a].
\end{equation} 
This restriction is testable using widely available parametric or non-parametric approaches for modeling the conditional expectations (e.g., \citep{racine2006testing, luedtke2019omnibus}). Such statistical testing may be a useful adjunct to assessments of the causal assumptions based on substantive knowledge. 

In practical terms, violations of the restrictions in displays \eqref{eq:independencies4} or \eqref{eq:observed_data_implications2} can occur, for example, when there are unmeasured common causes of trial participation $S$ and the outcome $Y$ (failure of condition \emph{A4}), when trial participation directly effects the outcomes (failure of condition \emph{A1}), or when there exists outcome-relevant variation in the treatments evaluated across trials (again, failure of condition \emph{A1}). Though formal statistical assessments of the restrictions are possible, they will often be challenging when $X$ is high-dimensional. Furthermore, formal assessments of the restrictions in \eqref{eq:independencies4} and \eqref{eq:observed_data_implications2} will in general not be able to pinpoint the specific identifiability condition whose failure explains the violation of the observed data restrictions. Thus, considerable background knowledge and scientific judgment will be needed to decide on the appropriateness of combining information across trials to learn about the target population.

\subsection{Identification of potential outcome means using the collection of trials}

When conditions \emph{A1} through \emph{A5} hold, the potential outcome mean in the target population can be identified using covariate, treatment, and outcome data from the collection of trials, and baseline covariate data from the target population. 

\begin{restatable}[Identification of potential outcome means]{theorem}{thmidentificationcollection}
\label{thm:identification_collection}
Under conditions A1 through A5, the potential outcome mean in the target population under treatment $a \in \mathcal A$, $\E[Y^a | R = 0 ]$, is identifiable by the observed data functional
\begin{equation} \label{eq:identification_collection_g}
    \begin{split}
  \psi(a) &\equiv \E\big[ \E[Y | X, R = 1, A = a ] \big| R = 0 \big],
    \end{split}
\end{equation}
which can be equivalently expressed as 
\begin{equation} \label{eq:identification_collection_w}
    \begin{split}
  \psi(a) &= \dfrac{1}{\Pr[R = 0 ]}  \E\left[   \dfrac{ I(R = 1, A = a) Y \Pr[R = 0| X] }{\Pr[R = 1| X] \Pr[A = a | X,  R = 1 ]  } \right].
    \end{split}
\end{equation}
\end{restatable}

The proof is given in Web Appendix A.

It also follows that, under the conditions of Theorem \ref{thm:identification_collection}, the average treatment effect in the target population comparing treatments $a$ and $a^\prime$ in $\mathcal A$ is also identifiable: $\E [Y^a - Y^{a^\prime} | R = 0] =  \E [Y^a | R = 0] -  \E [Y^{a^\prime} | R = 0] = \psi(a) - \psi(a^\prime)$.

%%%%%%%%%%%%%%%%%%%%%%%%%%%%%%%%%%%%%%%%%%%%%%%%%%%%%%
\subsection{Identification under weaker conditions}\label{subsec:wekaer_positivity}
%%%%%%%%%%%%%%%%%%%%%%%%%%%%%%%%%%%%%%%%%%%%%%%%%%%%%%

Practitioners of meta-analysis believe that by combining evidence from multiple randomized trials it should be possible to draw inferences about target populations that are -- in some vague sense -- broader than the population underlying each randomized trial. To give a concrete example: suppose that we have data from two trials of the same medications, one recruiting individuals with mild disease and the other recruiting individuals with severe disease. Suppose also that our target population includes some individuals with mild and some with severe disease and that effectiveness varies by disease severity. Intuition suggests that if inferences are ``transportable'' from each of the two trials to each subset of the target population defined by disease severity, then it should be possible to draw some conclusions about treatment effectiveness in the target population. Yet, in this setting, condition \emph{A5} is grossly violated (e.g., individuals with mild disease have zero probability of participating in one of the two trials). In the following two sub-sections, we give identification results that rely on weaker identifiability conditions and capture the intuition that when combining evidence from multiple trials we can draw inferences about target populations that are broader than the population underlying each randomized trial.

%%%%%%%%%%%%%%%%%%%%%%%%%%%%%%%%%%%%%%%%%%%%%%%%%%%%%%
\subsubsection{Weakening the positivity conditions of Theorem \ref{thm:identification_collection}}
%%%%%%%%%%%%%%%%%%%%%%%%%%%%%%%%%%%%%%%%%%%%%%%%%%%%%%

Suppose that the independence conditions $Y^a \indep R|X$ and $Y^a \indep A | (X, R=1)$ hold, taken either as primitive conditions or as implications of identifiability conditions \emph{A2} and \emph{A4} (as was done for Theorem \ref{thm:identification_collection}). Furthermore, consider the following positivity conditions: 

\vspace{0.1in}
\noindent
\emph{A3$^*$. Positivity of the probability of treatment in the collection of trials:} for each treatment $a \in \mathcal A$, if $f(x,R=1) \neq 0$, then $\Pr[A = a | X = x, R = 1] > 0$.

\vspace{0.1in}
\noindent
\emph{A5$^*$. Positivity of the probability of participation in the collection of trials:} if $f(x,R=0) \neq 0$, then $\Pr[R = 1 | X = x] > 0$.

It is worth noting that these positivity conditions are weaker than the corresponding conditions \emph{A3} and \emph{A5}, in the sense that \emph{A3$^*$} is implied by, but does not imply, \emph{A3}; and \emph{A5$^*$} is implied by, but does not imply, \emph{A5}. The difference between the two sets of assumptions is practically important: informally, condition \emph{A3} requires every treatment to be available in every trial and for every covariate pattern that can occur in that trial, and condition \emph{A5} requires every covariate pattern that can occur in the target population to be represented in every trial. In contrast, condition \emph{A3$^*$} requires every treatment to be available in the aggregated collection of trials and for any covariate pattern that can occur in the collection (but not necessarily in every trial); condition \emph{A5$^*$} requires every covariate pattern that can occur in the target population to occur in the aggregated collection of trials (but not necessarily in every trial).

Using these modified conditions, we obtain the following result (we give the formal arguments in Web Appendix A and Web Appendix B):

\begin{restatable}[Identification of potential outcome means under weaker positivity conditions]{theorem}{thmidentificationcollectionweakerpos}
\label{thm:identification_collection_weakerpos}
If $Y^a \indep R|X$ and $Y^a \indep A | (X, R=1)$ and conditions A1, A3$^*$, and A5$^*$ hold, then the potential outcome mean in the target population under treatment $a \in \mathcal A$, $\E[Y^a | R = 0 ]$, is identifiable by $\psi(a)$.
\end{restatable}

%%%%%%%%%%%%%%%%%%%%%%%%%%%%%%%%%%%%%%%%%%%%%%%%%%%%%%
\subsubsection{Identification under weaker overlap and exchangeability conditions}
%%%%%%%%%%%%%%%%%%%%%%%%%%%%%%%%%%%%%%%%%%%%%%%%%%%%%%

We will now show that the exchangeability conditions for potential outcome means can also be relaxed, alongside the positivity conditions; to do so, we need to introduce some additional notation in order to keep track of the subsets of the collection of trials where different covariate patterns can occur.  

Let $\mathcal X_j$ for $j \in \mathcal \{0, 1, \ldots, m\}$ denote the support of the random vector $X$ in the subset of the population with $S = j$. That is, $\mathcal X_j \equiv \{ x: f(x | S = j) > 0 \}$. For each covariate pattern $X = x$ that can occur in the collection of trials, that is, for each $x \in \bigcup\limits_{s \in \mathcal S} \mathcal X_s$, define $\mathcal S_x$ as the subset of trials in $\mathcal S$ such that $x$ belongs in their support. That is, $\mathcal S_x \equiv \{ s: s \in \mathcal S, x \in \mathcal X_s \}$. Intuitively, $\mathcal S_x$ denotes the subset of trials in the collection $\mathcal S$ where the covariate pattern $X = x$ can occur. Using this additional notation, consider the following identifiability conditions:

\vspace{0.1in}
\noindent
\emph{A4$^{\dagger}$. Exchangeability in mean over $S$:} For every $x$ such that $f(x, S =0) \neq 0$ and every $s \in \mathcal S_x$, $\E[Y^a | X = x, S = 0 ] = \E [Y^a | X = x, S = s]$.

\vspace{0.1in}
\noindent
\emph{A5$^{\dagger}$. Overlap of the collection $\mathcal S$ with the target population:} $\bigcup\limits_{s \in \mathcal S} ( \mathcal X_s \cap \mathcal X_0 ) = \mathcal X_0$.

\noindent
Using the above two conditions in the place of conditions \emph{A4} and \emph{A5} still allows for identification of the potential outcome means, using the following result. 

\begin{restatable}[Identification of potential outcome means under weaker conditions]{theorem}{thmidentificationwekapos}
\label{thm:identification_pos}
Under identifiability conditions A1 through A3, A4$^\dagger$, and A5$^\dagger$, for every $a \in \mathcal A$, $\E[Y^a | R = 0]$ is identifiable by $$ \phi(a) \equiv  \int\E \big[Y | X = x, I(S \in \mathcal S_x) = 1, A = a \big] f(x | S = 0) dx. $$
\end{restatable}

The proof is given in Web Appendix B.

In simple cases, as in our hypothetical example of effect modification by disease severity where the positivity violation was due to a single binary covariate, the above result simply suggests that the conditional mean of the outcome $Y$ can be modeled separately in each group of trials defined by that binary covariate. For more complicated violations of positivity condition \emph{A5}, especially those involving multiple continuous covariates, the above result is mostly useful, not for suggesting a particular modeling strategy, but for showing that some degree of interpolation when modeling the conditional expectation may be relatively benign, when there is adequate overlap between the collection of trials and the target population, in the sense of condition \emph{A5}$^\dagger$.

%%%%%%%%%%%%%%%%%%%%%%%%%%%%%%%%%%%%%%%%%%%%%%%%%%%%%%
\subsection{Identification under exchangeability in effect measure}\label{subsec:id_transport_in_measure}
%%%%%%%%%%%%%%%%%%%%%%%%%%%%%%%%%%%%%%%%%%%%%%%%%%%%%%

Up to this point, we have examined conditions that are sufficient for identifying both potential outcome means and average treatment effects. If interest is restricted to average treatment effects, identification is possible under a weaker condition of exchangeability in effect measure over $S$ given baseline covariates, essentially, a requirement that conditional average treatment effects, but not necessarily potential outcome means, can be transported from each trial in the collection $S$ to the target population; see \citet{dahabreh2018generalizing, dahabreh2020transportingStatMed} for a similar argument in the context of transporting inferences from a single trial. Under this weaker condition, however, the individual potential outcome means are not identifiable. Specifically, consider the following condition: 

\vspace{0.1in}
\noindent
\emph{A4$^{\ddagger}$. Exchangeability in effect measure over $S$:} for every pair of treatments $a$ and $a^\prime$, with $a \in \mathcal A$ and $a^\prime \in \mathcal A$, for every $s \in \mathcal S$, and for every $x$ such that $f(x, S = s) \neq 0$, we have $\E[Y^a - Y^{a^\prime} | X = x, S = s] = \E[Y^a - Y^{a^\prime} | X = x, R = 0 ]$.

It is easy to see that this condition, combined with conditions \emph{A1} through \emph{A3}, and \emph{A5}, implies that the observed data trial-specific conditional mean difference comparing treatments $a \in \mathcal A$ and $a^\prime \in \mathcal A$, $\E[Y|X,S=s,A=a] - \E[Y|X,S =s, A= a^\prime] \mbox{ does not vary over } s \in \mathcal S$ for all $X$ values with positive density in the target population. Using this implication, we now give an identification result for the average treatment effect under the weaker condition of exchangeability in effect measure.

\begin{restatable}[Identification under exchangeability in effect measure]{theorem}{thmidentificationdiff}
\label{thm:identification_diff}
Under conditions A1 through A3, A4$^{\ddagger}$, and A5,  the average treatment effect in the target population comparing treatments $a \in \mathcal A$ and $a^\prime \in \mathcal A$, $\E[Y^a - Y^{a^\prime} | R = 0]$, is identifiable by the observed data functional $\rho(a,a^\prime) \equiv \E \big[ \tau(a, a^\prime; X) \big| R = 0 \big]$, where $\tau(a, a^\prime; X) \equiv \E[Y | X, S = s, A = a] -  \E[Y | X, S=s, A = a^\prime]$ does not vary over $s \in \mathcal S$. Furthermore, $\rho(a,a^\prime)$ can be re-expressed as $\rho(a,a^\prime) = \dfrac{1}{\Pr[R = 0 ]} \E \left[ w(a, a^\prime; R,X,S,A)  Y  \right]$, where 
\begin{equation*}
    w(a, a^\prime; R,X,S,A) \equiv \left( \dfrac{I(R = 1, A = a)}{\Pr[A = a | X, S, R = 1]} - \dfrac{I(R = 1, A = a^\prime)}{\Pr[A = a^\prime | X, S, R = 1]}  \right) \dfrac{ \Pr[R = 0 | X]}{\Pr[R = 1 | X]} .
\end{equation*}
\end{restatable}

The proof is given in Web Appenidx C. Note in passing that positivity condition \emph{A5} in Theorem \ref{thm:identification_diff} can be relaxed in a way analogous to what was done in Theorem \ref{thm:identification_pos}. 

Because the potential outcome means $\E[Y^a | R = 0]$ in the target population often are of inherent interest, in the rest of this paper we focus on estimating the identifying functionals $\psi(a)$, for $a \in \mathcal  A$, following Theorems \ref{thm:identification_collection} and \ref{thm:identification_collection_weakerpos}.

%%%%%%%%%%%%%%%%%%%%%%%%%%%%%%%%%%%%%%%%%%%%%%%%%%%%%%%%%%%%%%%%%%%%%%%%%%%%%%
\section{Remarks on the sampling model}\label{sec:sampling_model}
%%%%%%%%%%%%%%%%%%%%%%%%%%%%%%%%%%%%%%%%%%%%%%%%%%%%%%%%%%%%%%%%%%%%%%%%%%%%%%

We now revisit the assumed sampling scheme that we described in Section \ref{sec:data_causal_quant} and introduce some notation to help us distinguish between the \emph{population sampling model}, where data are obtained by simple random sampling from a common super-population, and a \emph{biased sampling model}, where data are obtained by stratified random sampling from the target population and the populations underlying the trials, with unknown and possibly variable sampling fractions.  

Under a nonparametric model for the observed data, the density of the law of the observable data $O = (X,S,R,A,Y)$, using $p$ to generically denote densities, can be written as $p(r,x,s,a,y) = p(r) p(x|r) p(s|x, r) p(a|r,x,s) p(y|r, x, s, a)$. By definition, $S$ contains all the information contained in $R$, which implies that $p(a|r,s,x) = p(a|s,x)$. Furthermore, by \eqref{eq:independencies4}, $p(y|r= 1, x, a) = p(y|r=1, s, x, a)$; and the equivalence $\{R=0, S=0\} \iff \{S=0\} $ gives $p(y|r= 0, x, a) = p(y|r=0, s, x, a)$; thus, $p(y|r, x, a) = p(y|r, s, x, a)$. We conclude that, under simple random sampling, the density of the observable data would be $$p(r,x,s,a,y) = p(r) p(x|r) p(s|x,r) p(a|x,s) p(y|r, x, a).$$ This density can be viewed as reflecting a \emph{population sampling model}.

The data collection approach described in Section \ref{sec:data_causal_quant}, however, induces a \emph{biased sampling model} \citep{bickel1993efficient}. What we mean here is that the sampling fraction from the sub-population underlying each randomized trial ($S = 1, \ldots, m$) and the target population sample ($R = S = 0$) is not under the control of the investigators (and, typically, unknown to them) and reflects the particular circumstances (e.g., recruiting practices) of how sampling was conducted (e.g., convenience sampling is typical in randomized trials). Thus, in the data, the ratios $\dfrac{n_j}{ n}$, for $j \in \{0,1, \ldots, m\}$, do not reflect the population probabilities of belonging to the subset of the population with $S = j$, because the sampling fraction from each subset is unknown (and possibly variable between subsets). As a technical condition, we require that as $n \longrightarrow \infty$, $\dfrac{n_j}{ n} \longrightarrow \pi_j > 0$, for $j \in \{0,1, \ldots, m\}.$ Nevertheless, under the biased sampling model, the limiting values, $\pi_j$, are not necessarily equal to the super-population probabilities under the population sampling model.

Under this more plausible sampling model, the density of the observable data can be written as
\begin{equation*}
    \begin{split}
       q(r,x,s,a,y) &=  q(r) q(x|r) q(s|x,r) p(a|x,s) p(y|r, x, a) \\
        &= q(r)\{ q(x|r = 1) \}^r \{ q(x|r = 0) \}^{1-r} q(s|x,r) p(a|x,s) p(y|r, x, a) \\
        &= q(r)\{ q(x|r = 1) \}^r \{ p(x|r = 0) \}^{1-r} q(s|x,r) p(a|x,s) p(y|r, x, a).
    \end{split}
\end{equation*}
Here, $q(r) \neq p(r)$ and $q(s|x, r) \neq p(s|x, r)$ to reflect the fact that the sampling fractions for participants in the trials and members of the target population are not equal to the super-population probabilities under the population sampling model, but instead depend on the complex processes of randomized trial design and conduct (i.e., the biased sampling model). For the same reason, $ q(x|r) \neq  p(x|r) $, because in general we do not expect $ q(x|r=1) =  p(x|r=1) $; the distribution of $X$ in the collection of trials depends on the sampling from the different population underlying each trial. Nevertheless, we expect $p(x|r = 0)  = q(x|r = 0) $ because we take a simple random sample from the target population under the biased sampling model. In contrast, the terms $p(a|x,s)$ and $p(y|a,x,r)$ are the same in both densities, reflecting the property of stratified (by $S$) random sampling of the population underlying each trial or the target population, and the independence condition in \eqref{eq:independencies4}.

The functional $\psi(a)$ in Theorems \ref{thm:identification_collection} and \ref{thm:identification_collection_weakerpos} only depends on components of $p(r,x,s,a,y)$ that are also identifiable under the distribution $q(r,s,x,a,y)$ induced by the biased sampling. Specifically, using the representation in (\ref{eq:identification_collection_g}), $\psi(a)$ depends on $p(y|r=1,x,a)$, and $p(x|r = 0) = p(x|s = 0)$, both of which are identifiable under the biased sampling model.

%%%%%%%%%%%%%%%%%%%%%%%%%%%%%%%%%%%%%%%%%%%%%%%%%%%%%%%%%%%%%%%%%%%%%%%%%%%%%%
\section{Estimation and inference} \label{sec:estimation}
%%%%%%%%%%%%%%%%%%%%%%%%%%%%%%%%%%%%%%%%%%%%%%%%%%%%%%%%%%%%%%%%%%%%%%%%%%%%%%

We now turn our attention to the estimation of the functional $\psi(a)$, for $a \in \mathcal A$. The results presented in this section apply to data obtained under the population sampling model $p(r,x,s,a,y)$ and the biased sampling model $q(r,x,s,a,y)$, because $\psi(a)$ is identifiable under both. In fact, the results in \citet{breslow2000semi} imply that influence functions for $\psi(a)$ under sampling from  $q(r,x,s,a,y)$ are equivalent to those under sampling from $p(r,x,s,a,y)$, but with densities from $q(r,x,s,a,y)$ replacing those under $p(r,x,s,a,y)$; see \citet{kennedy2015semiparametric} for a similar argument in the context of matched cohort studies and \citet{dahabreh2020transportingStatMed} in the context of transporting inferences from a single trial. In the remainder of the paper, we assume that analysts will be working under $q(r,x,s,a,y)$, because the biased sampling model is more realistic for applied meta-analyses.

\subsection{Proposed estimator} 

In Web Appendix D we show that the first-order influence function \citep{bickel1993efficient} of $\psi(a)$, for $a \in \mathcal A$, under the nonparametric model for the observable data, is
\begin{equation*}
  \begin{split}
    &\mathit\Psi^1_{q_0}(a) = \dfrac{1}{\pi_{q_0}}  \Bigg\{  \dfrac{I(R = 1, A = a) \Pr_{q_0}[R = 0 | X]}{ \Pr_{q_0}[R = 1 | X] \Pr_{q_0}[A =a | X, R = 1] } \Big\{ Y-  \E_{q_0}[Y | X, R = 1, A = a] \Big\}   \\
      &\quad\quad\quad\quad\quad\quad\quad\quad\quad+ I(R=0) \Big\{  \E_{q_0}[Y | X, R = 1,  A = a]  - \psi_{q_0}(a)\Big\}  \Bigg\},
  \end{split}
\end{equation*}
where $\pi_{q_0} = \Pr_{q_0}[R = 0]$ and the subscript $q_0$ denotes that all quantities are evaluated at the ``true'' data law. Specifically, we show that $\mathit\Psi^1_{q_0}(a)$ satisfies $\dfrac{\partial \psi_{q_t}(a)}{\partial t} \Bigg|_{t=0}  = \E[\mathit\Psi^1_{q_0}(a) u(O)]$, where $u(O)$ denotes the score of the observable data $O = (R, X, S, A, Y)$ and the left hand side of the above equation is the pathwise derivative of the target parameter $\psi(a)$. Theorem 4.4 in \citep{tsiatis2007} shows that $\mathit\Psi_{q_0}^1(a)$ lies in the tangent set; it follows (see, e.g., \citep{van2000asymptotic}, page 363) that $\mathit\Psi^1_{q_0}(a)$ is the efficient influence function under the nonparametric model (it is in fact the unique influence function under that model). For a more thorough discussion of semiparametric efficiency theory and precise definitions of pathwise derivatives and the tangent space we refer to Chapter 25 in \citep{van2000asymptotic}. Furthermore, in Web Appendix D we show that $\mathit\Psi^1_{q_0}(a)$ is the efficient influence function under semiparametric models incorporating the restriction $Y\indep S|(X, A = a,R)$ or where the probability of treatment conditional on covariates and trial participation status is known (e.g., if all trials have the same randomization probabilities).

The influence function above suggests the estimator
\begin{equation}\label{eq:estimator}
    \begin{split}
  \widehat \psi_{\text{\tiny aug}}(a) &= \dfrac{1}{ n \widehat \pi } \sum\limits_{i=1}^{n} \Bigg\{ I(R_i = 1, A_i = a) \dfrac{1 - \widehat p(X_i)}{\widehat p(X_i) \widehat e_a(X_i)}  \Big\{ Y_i - \widehat g_a(X_i) \Big\} + I(R_i = 0) \widehat g_a(X_i) \Bigg\},
  \end{split}
\end{equation}
where $\widehat \pi = n^{-1} \sum_{i=1}^n I(R_i = 0)$, $\widehat g_a(X)$ is an estimator for $\E[Y | X, R = 1, A = a]$, $\widehat e_a(X)$ is an estimator for $\Pr[A = a| X, R = 1]$, and $\widehat p(X)$ is an estimator for $\Pr[R = 1 | X]$. Note that $\widehat \psi_{\text{\tiny aug}}(a)$ involves data on $(R = 1, X, A, Y)$ from trial participants and data on $(R = 0, X)$ from the sample of the target population; thus, treatment and outcome data from the target population are not necessary. Furthermore, it is natural to estimate the average treatment effect in the target population for comparing treatments $a$ and $a^\prime$ in $\mathcal A$, that is, $\E[Y^a - Y^{a^\prime} | R = 0]$ , using the contrast estimator $\widehat \delta (a, a^\prime) = \widehat \psi_{\text{\tiny aug}}(a) - \widehat \psi_{\text{\tiny aug}}(a^\prime).$

In practical applications, analysts may be able to improve the performance of the estimator by normalizing the weights so that their sum is equal to the total number of observations in the target population sample \citep{dahabreh2018generalizing, dahabreh2019relation, dahabreh2020transportingStatMed}. This normalization, which was originally proposed for survey analyses \citep{hajek1971comment}, may be particularly useful when $\dfrac{1 - \widehat p(X)}{\widehat p(X) \widehat e_a(X)} $ is highly variable over the observations in the trials \citep{robins2007}.

\paragraph{Asymptotic properties of estimators for potential outcome means:} Let $g_{a}^{*}(X), e_a^*(X)$, and $p^{*}(X)$ denote the asymptotic limits (assumed to exist) of $\widehat g_{a}(X), \widehat e_a(X)$, and $ \widehat p (X)$, respectively. Finally, define $\widehat \gamma = \left\{ \frac{1}{n} \sum_{i=1}^n I(R_i = 0 ) \right\}^{-1}  = \widehat \pi^{-1}$ as an estimator for $\gamma^* = \Pr[R=0]^{-1}$.

For general functions $\gamma', g'_{a}(X), e'_{a} (X),$ and $p'(X)$ define 
\begin{equation*}
  \begin{split}
H(\gamma'&, g'_{a}(X), e'_{a} (X), p'(X)) =  \gamma' \Bigg\{ I(R = 0) g'_{a}(X) + I(R = 1, A =a ) \dfrac{  1 - p' (X) }{p'(X) e'_{a} (X)} \Big\{Y  -  g'_{a}(X)  \Big\}    \Bigg\}.
  \end{split}
\end{equation*}
Using notation from \citet{van1996weak}, define $\mathbb{P}_n\big(v(W)\big) = n^{-1} \sum_{i=1}^n v(W_i)$ and $\mathbb{G}_n(v(W)) = \sqrt{n}\big(\mathbb{P}_n(v(W)) - \E[v(W)]\big)$, for some function $v$ and a random variable $W$. Using this notation, $\widehat \psi_{\text{\tiny aug}}(a) = \mathbb{P}_n\big(H(\widehat \gamma, \widehat  g_{a}(X), \widehat e_{a} (X), \widehat p(X))\big)$.

To establish asymptotic properties of $\widehat \psi(a)$, we make the following assumptions:
\begin{enumerate}
\item[(i)] The sequence $H(\widehat \gamma, \widehat  g_{a}(X), \widehat e_{a} (X), \widehat p(X))$ and its limit $H(\gamma^*, g^*_{a}(X), e_{a}^* (X), p^*(X))$ fall in a Donsker class \citep{van1996weak}.
\item[(ii)] $\big|\big|H(\widehat \gamma, \widehat  g_{a}(X), \widehat e_{a} (X), \widehat p(X)) - H(\gamma^*, g^*_{a}(X), e_{a}^* (X), p^*(X))\big|\big|_2 \overset{a.s.}{\longrightarrow} 0.$
\item[(iii)] $ \E\big[H(\gamma^*, g^*_{a}(X), e_{a}^* (X), p^*(X))^2\big] < \infty$.
\item[(iv)] At least one of the following two assumptions holds: (a) $\widehat p(X) \overset{a.s.}{\longrightarrow} p^*(X) = \Pr[ R = 1 | X ]$ and $\widehat e_a(X) \overset{a.s.}{\longrightarrow} e_a^*(X) = \Pr[A=a|X, R=1]$; or (b) $\widehat g_{a}(X) \overset{a.s.}{\longrightarrow} g_{a}^{*}(X) = \E[Y|X, R=1, A = a]$, $p^*(X) \geq \varepsilon, \mbox{ and } e_a^*(X) \geq \varepsilon \mbox{ a.s.}$ for some $\varepsilon > 0.$ 
\end{enumerate}
The following Theorem gives the asymptotic properties of $\widehat \psi_{\text{\tiny aug}}(a)$; a detailed proof is given in Web Appendix E. 
\begin{restatable}{theorem}{thmestimationSA}
\label{thm:estimation_g_SA}
If assumptions \emph{(i)} through \emph{(iv)} hold, then 
\begin{enumerate}
\item $\widehat \psi_{\text{\tiny aug}}(a) \overset{a.s.}{\longrightarrow} \psi(a)$; and
\item $\widehat \psi_{\text{\tiny aug}}(a)$ has the asymptotic representation
\begin{equation}
\sqrt{n} \big(\widehat \psi_{\text{\tiny aug}}(a)  - \psi(a)\big) = \mathbb{G}_n\big(H(\gamma^*, g^*_{a}(X), e_{a}^* (X), p^*(X))\big) + Rem + o_P(1), \label{As-Rep}
\end{equation}
where $\mathbb{G}_n\big(H(\gamma^*, g^*_{a}(X), e_{a}^* (X), p^*(X))\big)$ is asymptotically normal and
\begin{align}
Rem  &\leq \sqrt{n} O_P\bigg( \Big(\big|\big|\Pr[R = 1|X] - \widehat p(X) \big|\big|_2 + \big|\big| \Pr[A=a|X, R=1] - \widehat e_a(X)\big|\big|_2 \Big)  \nonumber \\
&\quad \quad \quad \quad \quad \quad \quad \quad \times \big|\big| \widehat g_{a}(X) -  \E[Y|X, R=1, A=a] \big|\big|_2 \bigg). \label{As-Rep-2}
\end{align}
\end{enumerate}
\end{restatable}
If $\widehat  g_{a}(X), \widehat  p(X), \widehat e_{a} (X)$ and $g^*_{a}(X), p^*(X), e_{a}^* (X)$ are Donsker; $p^*(X)$ and $e_a^*(X)$ are uniformly bounded away from zero -- assumption (iv); and $g^*_{a}(X)$ and $Y$ are uniformly bounded, then $H(\widehat \gamma,\widehat  g_{a}(X), \widehat e_{a} (X), \widehat p(X))$ and its limit are Donsker \citep{kosorok2008introduction}. 

Assumptions (i), (ii), and (iii) are standard assumptions used to show asymptotic normality of M-estimators \citep{van2000asymptotic}. Assumption (i) restricts the flexibility of the models used to estimate the nuisance parameters (and their corresponding limits). But, it still covers a wide range of commonly used estimators such as parametric Lipschitz classes and VC classes \citep{van1996weak}. For data-adaptive estimators, Donsker assumptions can be relaxed using sample splitting \citep{robins2008higher}.

Assumption (iv) indicates that the estimator $\widehat \psi_{\text{\tiny aug}}(a)$ is doubly robust, in the sense that it converges almost surely to $\psi(a)$ when either (1) the model for the conditional outcome mean $\E[Y | X, R = 1, A = a]$ is correctly specified, so that $\widehat g_{a}(X)$ converges to the true conditional expectations almost surely; or (2) the model for the probability of participation in any trial $\Pr[R = 1 | X]$ and the model for the probability of treatment assignment $\Pr[A=a|X, R=1]$ are correctly specified, so that $\widehat p(X)$ and $\widehat e_a(X)$ converge to the true conditional probabilities almost surely. When the treatment assignment is the same for all the trials $\Pr[A=a|X, R=1]$ is known and when the treatment assignment only depends on a few categorical covariates $\Pr[A=a|X, R=1]$ is easy to estimate consistently. But, when the treatment assignment mechanism depends on continuous covariates, as would be the case in pooled analyses of observational studies, estimating $\Pr[A=a|X, R=1]$ may be more challenging.

The asymptotic representation given by \eqref{As-Rep} gives several useful insights into how estimation of the nuisance parameters affect the asymptotic distribution of $\widehat \psi_{\text{\tiny aug}}(a)$. By the central limit theorem, the first term on the right hand side of \eqref{As-Rep} is asymptotically normal. Hence, the asymptotic distribution of the estimator relies on the behavior of the term given in equation \eqref{As-Rep-2}. If both the pair of nuisance parameters $\widehat p(X)$ and $\widehat g_{a}(X)$ and the pair of nuisance parameters $\widehat e_a(X)$ and $\widehat g_{a}(X)$ converge combined at a rate fast enough such that term \eqref{As-Rep-2} is $o_P(1)$, then the estimator is $\sqrt{n}$-consistent, asymptotically normal, and has asymptotic variance equal to the variance of the efficient influence function. % This happens, for example, if a) all the estimators $\widehat p(X)$, $\widehat e_a(X)$, and $\widehat g_{a}(X)$ are estimated using correctly specified generalized additive models, or b) if either $\widehat g_{a}(X)$ is estimated using a correctly specified parametric model and both $\widehat p(X)$, $\widehat e_a(X)$ are consistent or $\widehat p(X)$, $\widehat e_a(X)$ are both estimated using correctly specified parametric models and $\widehat g_a(X)$ is consistent.

The asymptotic representation gives the rate of convergence result
\begin{align*}
\big|\big|\widehat \psi_{\text{\tiny aug}}(a)  - \psi(a)\big|\big|  &= O_P\bigg( \frac{1}{\sqrt{n}} + \Big(\big|\big|\Pr[R = 1|X] - \widehat p(X) \big|\big|_2 + \big|\big| \Pr[A=a|X, R=1] - \widehat e_a(X)\big|\big|_2 \Big)  \nonumber \\
&\quad \quad \quad \quad \quad \quad \quad \quad \times \big|\big| \widehat g_{a}(X) -  \E\big[Y|X, R=1, A=a] \big|\big|_2 \bigg)
\end{align*}
Thus, if both the combined rate of convergence of $\widehat p(X)$ and $\widehat g_{a}(X)$ and the combined rate of convergence of $\widehat e_a(X)$ and $\widehat g_{a}(X)$ are at least $\sqrt{n}$, then $\psi_{\text{\tiny aug}}(a)$ is $\sqrt{n}$ convergent. %The combined rate of convergence of the pair $(\widehat p(X), \widehat e_{a}(X))$ and $\widehat g_{a}(X)$ is $\sqrt{n}$, for example, if one of the models for the pair $(\widehat p(X), \widehat e_{a}(X))$ or  $\widehat g_{a}(X)$ is estimated using a correctly specified parametric model and the other is estimated using a misspecified parametric model.

\subsection{Inference}

To construct Wald-style confidence intervals for $\psi(a)$, when using parametric models, we can easily obtain the sandwich estimator \citep{stefanski2002} of the sampling variance for $\widehat\psi(a)$. Alternatively, we can use the non-parametric bootstrap. If the remainder term in equation \eqref{As-Rep} of Theorem \ref{thm:estimation_g_SA} is $o_P(1)$, then the estimator $\widehat \psi_{\text{\tiny aug}}(a)$ is asymptotically normally distributed. The asymptotic normality combined with Assumption (iii) ensure that the confidence intervals calculated using  the non-parametric bootstrap asymptotically have the correct coverage rate.

%%%%%%%%%%%%%%%%%%%%%%%%%%%%%%%%%%%%%%%%%%%%%%%%%%%%%%%%%%%%%%%%%%%%%%%%%%%%%%
%%%%%%%%%%%%%%%%%%%%%%%%%%%%%%%%%%%%%%%%%%%%%%%%%%%%%%%%%%%%%%%%%%%%%%%%%%%%%%
\section{Simulation study}
%%%%%%%%%%%%%%%%%%%%%%%%%%%%%%%%%%%%%%%%%%%%%%%%%%%%%%%%%%%%%%%%%%%%%%%%%%%%%%
%%%%%%%%%%%%%%%%%%%%%%%%%%%%%%%%%%%%%%%%%%%%%%%%%%%%%%%%%%%%%%%%%%%%%%%%%%%%%%

We performed simulation studies to evaluate the finite sample performance of the augmented estimators described in the previous section and compare them with alternative approaches. 

\subsection{Data generation}

The data generation process involved six steps: generation of covariates, selection for trial participation, sampling of individuals from the target population, allocation of trial participants to specific trials, random treatment assignment, and potential/observed outcomes. Simulations contained either $n=10,000$ or $n=100,000$ individuals, including both trial participants and the sample of the target population.

\vspace{0.1in}
\noindent
\emph{1. Covariates:} Three covariates $X_1$, $X_2$, $X_3$ for each individual in the target population were drawn from a mean-zero multivariate normal distribution with all marginal variances equal to 1 and all pairwise correlations equal to 0.5.

\vspace{0.1in}
\noindent
\emph{2. Selection for trial participation:} We considered three trials, $\mathcal S = \{1,2,3\}$, a reasonable number given the requirement that all trials have examined the same treatments. We examined scenarios where the total number of trial participants,  $\sum_{s=1}^{3}n_s$, was 1,000, 2,000, or 5,000. We ``selected'' observations for participation in any trial using a logistic-linear model, $R \sim \text{Bernoulli} ( \Pr[R = 1| X ] )$ with $\Pr[R = 1| X ] = \dfrac{\text{exp}(\beta X^T) }{1+ \text{exp}(\beta X^T)}$, $X = (1, X_{1}, \ldots, X_{3})$, $\beta = (\beta_0, \ln(2), \ln(2),\ln(2) )$, where we solved for $\beta_0$ to result (on average) in the desired total number of trial participants given the total cohort sample size (exact numerical values for $\beta_0$ are available for all scenarios in the code to reproduce the simulations in Web Appendix H).

\vspace{0.1in}
\noindent
\emph{3. Sampling of individuals from the target population:} We used baseline covariate data from all remaining non-randomized individuals in the sample, $n_0 = n - \sum_{s=1}^{3}n_s$; this corresponds to taking a census of the non-randomized individuals in the simulated cohort and treating them as the sample from the target population.

\vspace{0.1in}
\noindent
\emph{4. Allocation of trial participants to specific trials:} We allocated trial participants ($R=1$) to one of the three randomized trials in $\mathcal S$ using a multinomial logistic model,  $S | (X,  R = 1)  \sim \text{Multinomial}\left(p_1, p_2, p_3; \sum_{s=1}^{3}n_s \right)$, with $p_1 = \Pr[S= 1 | X, R = 1 ] = 1 - p_2 - p_3$, $p_2 = \Pr[S= 2 | X, R = 1 ] = \dfrac{e^{\xi X^T} }{1+ e^{\xi X^T} + e^{\zeta X^T} }$, and $p_3 =\Pr[S= 3| X, R = 1 ] = \dfrac{e^{\zeta X^T} }{1+ e^{\xi X^T} + e^{\zeta X^T} }$, where $\xi = (\xi_0, \ln(1.5), \ln(1.5), \ln(1.5))$ and $\zeta = (\zeta_0, \ln(0.75),\ln(0.75),\ln(0.75))$. We evaluated a scenario in which the trials had the same sample size and one where the sample size varied across trials. We used Monte Carlo methods to obtain intercepts $\xi_0$ and $\zeta_0$ that resulted in approximately equal-sized trials or in unequal-sized trials with a 4:2:1 ratio of samples sizes (exact numerical values for $\xi_0$ and $\zeta_0$ are available for all scenarios in the code to reproduce the simulations) \citep{robertson2021intercept}.

\vspace{0.1in}
\noindent
\emph{5. Random treatment assignment:}
We generated an indicator of unconditionally randomized treatment assignment, $A$, among randomized individuals. In one scenario the treatment assignment mechanism was marginally randomized and constant across trials, $A \sim \text{Bernoulli}(\Pr [A = 1 |S = s])$ with $ \Pr [A = 1 |S = s] = 1/2$. In the second scenario, the treatment assignment mechanism varied across trials and was marginally randomized, with probabilities $ \Pr [A = 1 |S = 1] = 1/2$, $ \Pr [A = 1 |S = 2] = 1/3$, and $ \Pr [A = 1 |S = 3] = 2/3$.

\vspace{0.1in}
\noindent
\emph{6. Outcomes:} We generated potential outcomes as $Y^a =   \theta^a X^T + \epsilon^a, \mbox{ for } a \in \{0,1\}$, $\theta^{0} = (1.5,1,1,1)$, and $\theta^{1} = (0.5, -1, -1, -1) $. In all simulations, $\epsilon^a$ had an independent standard normal distribution for $a=0,1$. We generated observed outcomes under consistency, such that $ Y = A Y^1 + (1 - A) Y^0.$

\subsection{Methods implemented and comparisons}\label{subsec:estimators_used_in_sim}

\emph{Estimators:} In each simulated dataset, we applied the estimators from Section \ref{sec:estimation}. Heuristically, the identification results in Theorems \ref{thm:identification_collection} and \ref{thm:identification_collection_weakerpos} suggest two non-augmented (and non-doubly robust) estimators: the g-formula estimator, 
\begin{equation*}
  \widehat \psi_{\text{\tiny g}}(a) = \Bigg\{\sum\limits_{i=1}^n I(R_i = 0) \Bigg\}^{-1} \sum\limits_{i=1}^{n}  I(R_i = 0) \widehat g_a(X_i),
\end{equation*}
and the weighting estimator
\begin{equation*}
  \widehat \psi_{\text{\tiny w}}(a) = \Bigg\{\sum\limits_{i=1}^n I(R_i = 0) \Bigg\}^{-1} \sum\limits_{i=1}^{n} I(R_i = 1, A_i = a) \dfrac{1 - \widehat p(X_i)}{\widehat p(X_i) \widehat e_a(X_i)}  Y_i  ,
\end{equation*}
where $\widehat g_{a}(X)$, $\widehat p (X)$, and $ \widehat e_{a} (X)$ were previously defined. The estimators $\widehat \psi_{\text{\tiny g}}(a)$ and $\widehat \psi_{\text{\tiny w}}(a)$ can be viewed as special cases of the augmented estimator $\widehat \psi_{\text{\tiny aug}}(a)$: we obtain $\widehat \psi_{\text{\tiny g}}(a)$ by identically setting $\widehat p(X)$ to one, and we obtain $\widehat \psi_{\text{\tiny w}}(a)$ by identically setting $\widehat g_{a}(X)$ to zero, in equation \eqref{eq:estimator}. It follows that $\widehat \psi_{\text{\tiny g}}(a)$ and $\widehat \psi_{\text{\tiny w}}(a)$ are not robust to misspecification of the models used to estimate $\widehat g_a(X)$, or $\widehat p(X) \times \widehat e_a(X)$, respectively \citep{dahabreh2020toward}.

\vspace{0.1in}
\noindent
\emph{Model specification:} All working models included main effects for the observable covariates $X_j$, $j=1,2,3$. Models for the probability of participation in the trial and models for the probability of treatment included main effects for all covariates. We used logistic regression to model the probability of participating in any trial ($R = 1$ vs 0). In our simulation scenarios, when the probability of treatment (the randomization ratio) varies across trials, a logistic regression model of treatment with main effects of the covariates (and no adjustment for trial $S$) is not correctly specified. To avoid misspecification of the model for the probability of treatment, we used the fact that $\Pr[A = a | X, R = 1] = \sum_{s=1}^m \Pr[A = a | X, S = s] \Pr[S = s | X, R = 1]$, and modeled the probability of treatment in each trial $\Pr[A = a | X, S = s]$ using logistic regression (separately in each trial) and the probability of participation in a particular trial among the collection of trials $\Pr[S = s | X, R = 1]$ using a multinomial logistic regression model (in the collection of trials). We refer to this approach as ``averaging the trial-specific treatment probabilities.'' For comparison, in the Appendix, we report results using a (misspecified) logistic regression model with main effects only. In our simulation setup, when the randomization ratio varies across trials, use of this model is expected to have some bias when using $ \widehat \psi_{\text{\tiny w}}(a)$, should not affect $\widehat \psi_{\text{\tiny g}}(a)$ at all, and should only affect the precision, but not induce bias for $\widehat \psi_{\text{\tiny aug}}(a)$, when the outcome model is correctly specified. Outcome models were fit separately in each treatment group using linear regression to allow effect modification by all covariates.   

\vspace{0.1in}
\noindent
\emph{Comparisons:} We compared the bias and variance of estimators over 10,000 runs for each scenario and each treatment. 

\subsection{Simulation results}

Tables \ref{tab:sim_res_selected_bias} and \ref{tab:sim_res_selected_variance} present selected results from our simulation study for scenarios with $n = 10,000$ and total trial sample size, $\sum_{s=1}^{3}n_s$, of 2000 individuals, when averaging the trial-specific treatment probabilities; complete simulation results are presented in Web Appendix F, including results for modeling the probability of treatment using logistic regression across all trials. Bias estimates for $\widehat \psi_{\text{\tiny aug}}(a)$ and $\widehat \psi_{\text{\tiny g}}(a)$ were near-zero for all estimators, regardless of how the probability of treatment was modeled; $\widehat \psi_{\text{\tiny w}}(a)$ also had near-zero bias except when the randomization ratio varied across trials and the probability of treatment was modeled using logistic regression fit across all trials (see Appendix Table S3). The low bias of the g-formula estimator $\widehat \psi_{\text{\tiny g}}(a)$ and the weighting estimator $\widehat \psi_{\text{\tiny w}}(a)$ (when the model for the probability of treatment was correctly specified) indirectly verify the double robustness property of the augmented estimator $\widehat \psi_{\text{\tiny aug}}(a)$; this is also supported by the near-zero bias of $\widehat \psi_{\text{\tiny aug}}(a)$ even when the randomization ratio varied across trials and the model for the probability of treatment was misspecified (Appendix Table S3). Across all scenarios, the sampling variance of the augmented estimator $\widehat \psi_{\text{\tiny aug}}(a)$ was larger than that of the outcome model-based estimator $\widehat \psi_{\text{\tiny g}}(a)$, but substantially smaller than that of the weighting estimator $\widehat \psi_{\text{\tiny w}}(a)$.

\begin{table}
\centering
\caption {Bias estimates based on 10,000 simulation runs; selected sample size scenarios; the probability of treatment in the collection of trials was estimated by averaging the trial-specific treatment probabilities.}\label{tab:sim_res_selected_bias} 
\begin{tabular}{@{}lccccccc@{}} 
\toprule
$a$ & $n$ & $\sum_{s=1}^{3}n_s$ & Balanced & TxAM varies & $\widehat \psi_{\text{\tiny aug}}(a)$ & $\widehat \psi_{\text{\tiny g}}(a)$ & $\widehat \psi_{\text{\tiny w}}(a)$ \\ \midrule
1 & 10000       & 2000        & Yes      & No           & 0.0022           & 0.0003   & -0.0011         \\
1 & 10000       & 2000        & Yes      & Yes          & 0.0004           & -0.0003  & -0.0026         \\
1 & 10000       & 2000        & No       & No           & -0.0006          & -0.0001  & -0.0082         \\
1 & 10000       & 2000        & No       & Yes          & 0.0011           & 0.0002   & -0.0067         \\
1 & 100000      & 2000        & Yes      & No           & -0.0003          & 0.0002   & -0.0104         \\
1 & 100000      & 2000        & Yes      & Yes          & -0.0002          & 0.0008   & 0.0102          \\
1 & 100000      & 2000        & No       & No           & -0.0012          & -0.0002  & -0.0074         \\
1 & 100000      & 2000        & No       & Yes          & 0.0009           & 0.0013   & 0.0039          \\ \midrule
0 & 10000       & 2000        & Yes      & No           & 0.0001           & 0.0008   & 0.0083          \\
0 & 10000       & 2000        & Yes      & Yes          & 0.0006           & 0.0004   & 0.0060          \\
0 & 10000       & 2000        & No       & No           & 0.0010           & -0.0003  & 0.0002          \\
0 & 10000       & 2000        & No       & Yes          & 0.0002           & 0.0005   & 0.0106          \\
0 & 100000      & 2000        & Yes      & No           & 0.0013           & 0.0005   & 0.0016          \\
0 & 100000      & 2000        & Yes      & Yes          & -0.0009          & -0.0000  & 0.0034          \\
0 & 100000      & 2000        & No       & No           & -0.0009          & -0.0003  & 0.0064          \\
0 & 100000      & 2000        & No       & Yes          & 0.0012           & 0.0006   & 0.0109          \\ \bottomrule
\end{tabular}
\vspace{0.5in}
\caption*{In the column titled \emph{Balanced}, \emph{Yes} denotes scenarios in which the trials had on average equal sample sizes; \emph{No}  denotes scenarios with unequal trial sample sizes. In the column titled  \emph{TxAM varies}, \emph{Yes} denotes scenarios in which the treatment assignment mechanism varied across trials; \emph{No} denotes scenarios in which the mechanism did not vary.}
\end{table}

\begin{table}
\centering
\caption {Variance estimates based on 10,000 simulation runs; the probability of treatment in the collection of trials was estimated by averaging the trial-specific treatment probabilities.}\label{tab:sim_res_selected_variance}
\begin{tabular}{@{}lccccccc@{}} 
\toprule
$a$ & $n$ & $\sum_{s=1}^{3}n_s$ & Balanced & TxAM varies & $\widehat \psi_{\text{\tiny aug}}(a)$ & $\widehat \psi_{\text{\tiny g}}(a)$ & $\widehat \psi_{\text{\tiny w}}(a)$ \\ \midrule
1 & 10000       & 2000        & Yes      & No           & 0.0105          & 0.0038  & 0.2893         \\
1 & 10000       & 2000        & Yes      & Yes          & 0.0084          & 0.0035  & 0.1945         \\
1 & 10000       & 2000        & No       & No           & 0.0101          & 0.0038  & 0.2167         \\
1 & 10000       & 2000        & No       & Yes          & 0.0088          & 0.0038  & 0.1925         \\
1 & 100000      & 2000        & Yes      & No           & 0.0144          & 0.0039  & 0.3316         \\
1 & 100000      & 2000        & Yes      & Yes          & 0.0120          & 0.0035  & 0.3610         \\
1 & 100000      & 2000        & No       & No           & 0.0143          & 0.0039  & 0.4028         \\
1 & 100000      & 2000        & No       & Yes          & 0.0128          & 0.0037  & 0.2925         \\ \midrule
0 & 10000       & 2000        & Yes      & No           & 0.0103          & 0.0039  & 0.1165         \\
0 & 10000       & 2000        & Yes      & Yes          & 0.0131          & 0.0042  & 0.1509         \\
0 & 10000       & 2000        & No       & No           & 0.0101          & 0.0039  & 0.1136         \\
0 & 10000       & 2000        & No       & Yes          & 0.0118          & 0.0040  & 0.1551         \\
0 & 100000      & 2000        & Yes      & No           & 0.0151          & 0.0039  & 0.1873         \\
0 & 100000      & 2000        & Yes      & Yes          & 0.0198          & 0.0043  & 0.2699         \\
0 & 100000      & 2000        & No       & No           & 0.0146          & 0.0038  & 0.1663         \\
0 & 100000      & 2000        & No       & Yes          & 0.0218          & 0.0040  & 0.3243         \\ \bottomrule
\end{tabular}
\vspace{0.5in}
\caption*{In the column titled \emph{Balanced}, \emph{Yes} denotes scenarios in which the trials had on average equal sample sizes; \emph{No}  denotes scenarios with unequal trial sample sizes. In the column titled \emph{TxAM varies}, \emph{Yes} denotes scenarios in which the treatment assignment mechanism varied across trials; \emph{No} denotes scenarios in which the mechanism did not vary.}
\end{table}

%%%%%%%%%%%%%%%%%%%%%%%%%%%%%%%%%%%%%%%%%%%%%%%%%%%%%%%%%%%%%%%%%%%%%%%%%%%%%%
%%%%%%%%%%%%%%%%%%%%%%%%%%%%%%%%%%%%%%%%%%%%%%%%%%%%%%%%%%%%%%%%%%%%%%%%%%%%%%
\section{Application of the methods to the HALT-C trial}
%%%%%%%%%%%%%%%%%%%%%%%%%%%%%%%%%%%%%%%%%%%%%%%%%%%%%%%%%%%%%%%%%%%%%%%%%%%%%%
%%%%%%%%%%%%%%%%%%%%%%%%%%%%%%%%%%%%%%%%%%%%%%%%%%%%%%%%%%%%%%%%%%%%%%%%%%%%%%

\subsection{Using the HALT-C trial to emulate meta-analyses}

\emph{The HALT-C trial data:} The HALT-C trial \citep{di2008prolonged} enrolled 1050 patients with chronic hepatitis C and advanced fibrosis who had not responded to previous therapy and randomized them to treatment with peginterferon alfa-2a ($a=1$) versus no treatment ($a=0$). Patients were enrolled at 10 research centers and followed up every 3 months after randomization. Here, we used the secondary outcome of platelet count at 9 months of follow-up as the outcome of interest; we report all outcome measurements as platelets$\times 10^3/$ml. We used the following baseline covariates: baseline platelet count, age, sex, previous use of pegylated interferon, race, white blood cell count, history of injected recreational drugs, ever received a transfusion, body mass index, creatine levels, smoking status,  previous use of combination therapy (interferon and ribavirin), diabetes, serum ferritin, hemoglobin, aspartate aminotransferase levels, ultrasound evidence of splenomegaly, and ever drinking alcohol. For simplicity, we restricted our analyses to patients with complete baseline covariate and outcome data ($n = 948$). 

\emph{Using the HALT-C trial to emulate a meta-analysis and evaluate the proposed methods:} We treated observations from the largest center with complete data in the HALT-C trial, with sample size $n_0 = 199$, as a sample from the target population, $S = R = 0$. We then used the data from the remaining 9 centers as a collection of ``trials,'' $\mathcal S = \{1,\ldots,9\}$, with a total sample size of $ \sum_{j=1}^9 n_j = 749$. Appendix Table S5 summarizes covariate information, stratified by $R$; Appendix Tables S6 and S7 summarize covariate information stratified by $S$. Because HALT-C used 1:1 randomized allocation, all trials had the same treatment assignment mechanism. Our goal was to transport causal inferences from the nine trials to the population represented by the target center. We were able to use the randomization in $R = 0$ to empirically benchmark the methods we propose using as a reference standard the unadjusted estimate of the average treatment effect in the target population. In view of how we configured the data for this illustrative analysis, our results should not be clinically interpreted. 

\emph{Methods implemented and comparisons:}\label{subsec:estimators_used_in_applied}
We implemented the augmented estimator $\widehat \psi_{\text{\tiny aug}}(a)$, along with the g-formula estimator $\widehat \psi_{\text{\tiny g}}(a)$ and the weighting estimator $\widehat \psi_{\text{\tiny w}}(a)$, defined in Section \ref{subsec:estimators_used_in_sim}. These estimators use covariate, treatment, and outcome data from the collection of trials $\mathcal S$ but only covariate data from the sample of the target population $R=0$. In these analyses we specified linear regression models for the conditional expectation of the outcome in each treatment group, and logistic regression models for the probability of trial participation (in any trial) and the probability of treatment among randomized individuals. All models included main effects  for the baseline covariates listed above. For comparison, we also used an unadjusted linear regression analysis to estimate the potential outcome means and average treatment effects in $R = 0$ only; this analysis is justified by randomization and uses treatment and outcome data from the target population.

To provide an assessment of the observed data independence condition in equation \eqref{eq:observed_data_implications2}, we used analysis of covariance (ANCOVA) to compare a linear regression model that included the main effects of baseline covariates, treatment, and all possible product terms between baseline covariates and treatment, against a linear regression model that additionally included the main effect of the trial indicators and all product terms between baseline covariates, treatment, and the trial indicators.

\subsection{Results in the emulated meta-analysis}

Table \ref{tab:HALT_C_results} summarizes results from the emulation of meta-analyses using HALT-C trial data. All transportability estimators, which only use baseline covariate data from $R = 0$ and covariate, treatment, and outcome data from $R = 1$, produced estimates that were very similar to the benchmark estimates from the unadjusted estimator that uses treatment and outcome data from the target center $R=0$. The three transportability estimators produced similar point estimates, which suggests that the transported inferences are not driven by model specification choices. The three transportability estimators also produced estimates similar to those of the benchmark unadjusted estimator in $R = 0$, which suggests (but does not prove) that the identifiability conditions needed for the different transportability estimators hold, at least approximately, and that model specification is approximately correct. The ANCOVA p-value of 0.425 did not indicate gross violations of the observed data implication in display \eqref{eq:observed_data_implications2}.

\begin{table}
\caption{Results from analyses using the HALT-C trial data.}\label{tab:HALT_C_results}
\centering
\scalebox{0.85}{
\begin{tabular}{@{}lcccc@{}}
\toprule
\textbf{Estimator}                  & $a = 1$         & $a = 0$         & \textbf{Mean difference}   \\ \midrule
$R=0$ (benchmark)    & 121.6 (111.0, 133.3) & 164.4 (151.2, 178.1) &  -42.8 (-60.2, -25.7)  \\ \midrule 
$\widehat \psi_{\text{\tiny aug}}(a)$ & 124.5 (116.3, 133.5) &	167.1 (157.3, 177.7) &	-42.6 (-51.2, -34.4)   \\
$\widehat \psi_{\text{\tiny g}}(a)$   & 124.6 (116.6, 133.0) &	167.1 (157.3, 177.3) &	-42.5 (-50.5, -34.9)   \\
$\widehat \psi_{\text{\tiny w}}(a)$  &  123.2 (113.6, 136.0) &	165.9 (152.6, 182.9) &	-42.7 (-60.1, -26.2) &  \\ \bottomrule
\end{tabular}
}
\vspace{0.5in}
\caption*{$R=0$ denotes analyses using data only from the sample of the target population. Numbers in parentheses are 95\% quantile-based confidence intervals from 10,000 bootstrap samples; 95\% normal confidence intervals from the sandwich estimator were similar. }
\end{table}

%%%%%%%%%%%%%%%%%%%%%%%%%%%%%%%%%%%%%%%%%%%%%%%%%%%%%%%%%%%%%%%%%%%%%%%%%%%%%%
\section{Discussion}
%%%%%%%%%%%%%%%%%%%%%%%%%%%%%%%%%%%%%%%%%%%%%%%%%%%%%%%%%%%%%%%%%%%%%%%%%%%%%%

We have described methods to transport causal inferences from a collection of randomized trials to a target population in which baseline covariate data are available but no experimentation has been conducted. Our results provide a solution to one of the leading problems in the field of ``evidence synthesis'': how to use a body of evidence consisting of multiple randomized trials to draw causal inferences about the effects of the interventions investigated in the trials for a new target population of substantive interest.

Traditional approaches to evidence synthesis, including meta-analysis, focus on modeling aspects of the distribution of study effects or producing statistically optimal summaries of available data \citep{higgins2009re, rice2018re}. In general, such analyses do not produce causally interpretable estimates of the effect of well-defined interventions on any target population because the contributions of individual studies to the meta-analysis summary are weighted by the precision of the study-specific estimates, without considering the relevance of the studies to any target population or the strong assumptions needed to combine information across studies to estimate treatment effects.

In contrast, our approach explicitly targets a well-defined population chosen on scientific or policy grounds, which may be different from the populations underlying the trials. In our experience, policy-makers who use evidence syntheses to inform their decisions are not interested in the populations represented by the completed trials and almost always have a different target population in mind. Thus, the statistical methods we propose can be viewed as a form of causally interpretable meta-analysis, consistent with the conceptual framework outlined in \citep{sobel2017causal}. An important aspect of our approach is the specification of the target population and the use of sample data from that population to estimate treatment effects. In practical applications, obtaining data from a target population can be challenging \citep{barker2021causally}, but we expect that challenge will become easier to address with increasing availability of administrative and registry data from populations of clinical or policy relevance. Another important aspect of our approach is the explicit statement of the assumptions needed for the causal interpretation of estimates produced by the estimators we propose. Though the assumptions are fairly strong even for the simple case we study here (point treatments with complete adherence and outcomes assessed at the end of the study), our approach naturally connect with the broad literature on causal inference and thus extensions to address time varying treatments, incomplete adherence, and longitudinal or failure-time outcomes are possible. We believe that making assumptions explicit allows reasoning about them and can be the basis of future work on sensitivity analyses. Future work may also address differential covariate measurement error in the randomized trials and the sample of the target population, systematically missing data (e.g., when some covariates are not collected in some trials), and variable selection methods (e.g., methods incorporating the independence in display \eqref{eq:independencies4}).

Our methods can are a generalization of methods for extending inferences from a single trial to a target population (e.g., \citet{cole2010, dahabreh2018generalizing, dahabreh2020transportingStatMed, rudolph2017}). They also relate to the growing literature on matching-adjusted indirect comparison (e.g., \citet{signorovitch2010comparative, signorovitch2012matching, cheng2020statistical, jackson2021alternative}). This literature has focused on the common situation where data on the target population are only available in summary form (e.g., estimates of the first moments of the covariates in $X$), typically does not use explicit causal (counterfactual) notation (with some recent exceptions, e.g., \citet{cheng2020statistical}), relies heavily on parametric assumptions (or other investigator-imposed constraints), and only applies to a single source trial (not multiple trials in meta-analysis) \citep{jackson2021alternative}. Our identification results suggest obvious extensions of matching-adjusted indirect comparison methods to the multiple trial case, provided investigators are willing to make the necessary causal assumptions. Conversely, methods from the literature on matching-adjusted indirect comparison and related matching methods \citep{jackson2021alternative} can be used to relax the requirement of individual participant data from the target population inherent in our approach (at the cost of reliance on additional parametric assumptions or reduced efficiency). It may also be interesting to combine our approach with recent extensions of meta-regression methods to incorporate matching-adjusted indirect comparisons \citep{phillippo2020multilevel}. 

We caution that our methods will be most useful when the available trials are quite similar (e.g., in terms of treatments, outcomes, and follow-up protocols), as may be the case in prospective research programs investigating a particular intervention or close replication efforts. It should come as no surprise that a strong degree of similarity is needed in order for a ``pooled'' analysis to produce estimates with a clear causal interpretation with respect to a particular target population; after all, our approach focuses on differences in participant characteristics across the populations underlying the trials and the target population, not variations in treatment. Extensions of our approach to address variations in the randomized treatments, as well as non-randomized co-interventions, would be valuable. We suggest, however, that summarizing bodies of evidence much larger and diverse than a few closely related studies, particularly when individual participant data are only available from a minority of sources, is best undertaken with a descriptive or predictive attitude, focusing on modeling the response surface of study effects \citep{rubin1992meta}, rather than on obtaining a single causally interpretable estimate. That said, we sincerely hope that future work will find ways to improve on this rather unambitious view.

\backmatter

%  This section is optional.  Here is where you will want to cite
%  grants, people who helped with the paper, etc.  But keep it short!

\section*{Acknowledgments}

This work was supported in part by Patient-Centered Outcomes Research Institute (PCORI) awards ME-1306-03758, ME-1502-27794, and ME-2019C3-17875; National Institutes of Health (NIH) grant R37 AI102634; and Agency for Healthcare Research and Quality (AHRQ) National Research Service Award T32AGHS00001. The content is solely the responsibility of the authors and does not necessarily represent the official views of PCORI, its Board of Governors, the PCORI Methodology Committee, NIH, or AHRQ. 

The data analyses in our paper used HALT-C data obtained from the National Institute of Diabetes and Digestive and Kidney Diseases (NIDDK) Central Repositories \\
(\url{https://repository.niddk.nih.gov/studies/halt-c/}; last accessed: January 18, 2021). This paper does not necessarily reflect the opinions or views of the HALT-C investigators, the NIDDK Central Repositories, or the NIDDK.\vspace*{-8pt}

%  If your paper refers to supplementary web material, then you MUST
%  include this section!!  See Instructions for Authors at the journal
%  website http://www.biometrics.tibs.org

\section*{Supplementary Materials}

Web Appendices A through G are available with this paper at the Biometrics website on Wiley Online Library. A copy of the Web Appendices, current as of February 4, 2022, is also available here: \href{https://www.dropbox.com/s/mtovcdyopqplhqi/BIOM2021038M.R1%20APPENDIX_causally_interpretable_meta_analysis__transportability_with_multiple_trials.pdf?dl=0}{web link}.

Code to reproduce our simulations and the HALT-C analyses, as well as an artificial data-set to illustrate the application of the methods are provided on GitHub: \\
\href{https://github.com/serobertson/EfficientCausalMetaAnalysis}{https://github.com/serobertson/EfficientCausalMetaAnalysis}.

\vspace*{-8pt}

%  Here, we create the bibliographic entries manually, following the
%  journal style.  If you use this method or use natbib, PLEASE PAY
%  CAREFUL ATTENTION TO THE BIBLIOGRAPHIC STYLE IN A RECENT ISSUE OF
%  THE JOURNAL AND FOLLOW IT!  Failure to follow stylistic conventions
%  just lengthens the time spend copyediting your paper and hence its
%  position in the publication queue should it be accepted.

%  We greatly prefer that you incorporate the references for your
%  article into the body of the article as we have done here 
%  (you can use natbib or not as you choose) than use BiBTeX,
%  so that your article is self-contained in one file.
%  If you do use BiBTeX, please use the .bst file that comes with 
%  the distribution.

\bibliographystyle{biom} 
\bibliography{biblio}

%\appendix

%  To get the journal style of heading for an appendix, mimic the following.

%\section{}
%\subsection{Title of appendix}

%Put your short appendix here.  Remember, longer appendices are possible when presented as Supplementary Web Material.  Please review and follow the journal policy for this material, available under Instructions for Authors at \texttt{http://www.biometrics.tibs.org}.

\label{lastpage}

\end{document}